\documentclass[11pt]{article}

\usepackage{authblk}   

\usepackage[a4paper,margin=1in]{geometry}
\usepackage{amsmath,amssymb,amsthm}
\usepackage{graphicx}
\usepackage{subcaption}
\usepackage{booktabs}
\usepackage{hyperref}
\usepackage{physics}
\usepackage{xcolor}
\usepackage{microtype}

\newtheorem{proposition}{Proposition}
\newcommand{\figdir}{figures}
\usepackage{tikz}
\usetikzlibrary{arrows.meta,positioning,calc,decorations.pathmorphing}
\usetikzlibrary{shapes}
\newtheorem{definition}{Definition}

\title{Adversarial Stress Tests for Quantum Certification}

\author[1]{Veronica Sanz}
\author[2]{Augusto Smerzi}

\affil[1]{Instituto de Física Corpuscular (IFIC), Universidad de Valencia--CSIC, E-46980 Valencia, Spain}
\affil[2]{College of Engineering Physics, Shenzhen Technology University, Shenzhen 518118, China}

\begin{document}
\maketitle

\begin{abstract}
We develop a practical framework for semi-device-independent (SDI)
certification under operational deviations from the ideal protocol model.
Apparent violations of classical benchmarks need not signal genuinely
non-classical behaviour; they can arise from misalignment between
(i) the scoring rule,
(ii) the finite-sample statistical bound applied to that score, and
(iii) the operational model realised in the experiment, including
bias, memory, drift, and selection effects.

We formalise a protocol-agnostic alignment principle based on a
martingale-safe lower confidence bound and an operationally consistent
effective classical ceiling. This yields a quantitative diagnostic,
the \emph{robustness gap}
$\Delta_{\mathrm{rob}} = S_{\mathrm{low}} - S_{C,\mathrm{eff}}$,
which separates statistical fluctuations from structural modelling
errors. Statistical deviations vanish asymptotically, whereas model
misalignment can produce persistent false certification unless the
benchmark is corrected.

Using the $2\!\to\!1$ random access code as a minimal SDI testbed,
we show that postselection can inflate conditional scores, whereas
unconditional scoring restores the correct operational meaning of the
witness. We further show that adaptive learning-based classical agents
do not enlarge the admissible classical set; rather, they recover the
effective classical ceiling implied by the operational model.

The resulting framework provides a systematic diagnostic for
certification in realistic quantum communication and measurement
settings with embedded classical control, adaptive processing, and
nonideal data acquisition.\end{abstract}


\section{Introduction}
\label{sec:intro}

Certification protocols are central to quantum information processing
and quantum-enhanced technologies. In semi-device-independent (SDI)
scenarios, operational constraints --- such as bounded dimension or
restricted communication --- enable performance guarantees without
full device characterisation. A core task is to determine, from finite
experimental data, whether an observed score certifies performance
beyond what is achievable by classical strategies consistent with
those assumptions.

Such certification statements rely implicitly on an operational
data model. Standard analyses typically assume independent and
identically distributed (IID) inputs, absence of selection effects,
and memoryless device behaviour. In realistic implementations,
however, these idealisations are rarely exact. Input settings may be
biased, temporal correlations may arise, and data may undergo
postselection or adaptive preprocessing. These deviations do not
necessarily signal a breakdown of the underlying physical principle;
rather, they modify the effective classical optimisation problem
against which performance must be evaluated.

This leads to a structural question: how can one distinguish genuine
supra-classical behaviour from artefacts induced by operational model
mismatch? In particular, how should scoring rules, statistical
confidence bounds, and classical reference values be defined so that
certification claims remain valid under realistic deviations?
Importantly, statistical fluctuations vanish with increasing sample
size, whereas model misalignment does not; without explicit alignment,
false certification can persist even asymptotically.

To our knowledge, this alignment problem has not been formulated
systematically in SDI certification. Existing analyses typically focus
on deriving classical bounds under idealised assumptions, or on
establishing finite-sample concentration guarantees once the
operational model is fixed. Operational deviations such as biased
sampling, selection effects, or adaptive classical control are often
treated as implementation imperfections, but not as part of the
benchmark-definition problem itself. This work isolates that missing
consistency layer and treats certification validity as an operational
alignment condition linking score, statistical bound, and classical
reference.

In this work we develop a minimal framework to address this issue.
Our starting point is an alignment principle: certification validity
requires consistency between (i) the score used to summarise
experimental data, (ii) the finite-sample statistical bound used to
certify that score, and (iii) the classical benchmark appropriate to
the operational data-generating process. Misalignment between these
elements can produce artificial ``violations'' of classical limits,
even when the underlying physical behaviour remains entirely classical.

We formalise a protocol-agnostic diagnostic quantity, the
\emph{robustness gap},
\begin{equation}
    \Delta_{\mathrm{rob}} = S_{\mathrm{low}} - S_{C,\mathrm{eff}},
\end{equation}
where $S_{\mathrm{low}}$ is a martingale-safe lower confidence bound
on the operational score and $S_{C,\mathrm{eff}}$ is the effective
classical ceiling defined under the same operational model.
A certification claim is supported only when
$\Delta_{\mathrm{rob}} > 0$ at a prescribed confidence level.
This construction cleanly separates statistical effects from
structural modelling assumptions.

To illustrate the framework in a controlled setting, we use the
minimal $2\!\to\!1$ random access code (RAC) as a testbed.
We analyse the impact of input bias, temporal correlations,
selection effects, and adaptive classical strategies.
We show that conditioning on ``kept'' events can inflate empirical
scores and generate apparent supra-classical performance when
compared to an inappropriate benchmark. By contrast, an unconditional
scoring rule combined with martingale-safe bounds restores robust
behaviour. Adaptive learning strategies are used only as constructive
stress tests: they recover at most the effective classical ceiling
implied by the operational model, rather than creating a new loophole.

Although we focus on a simple SDI scenario for transparency,
the framework is not protocol-specific. The alignment principle and
robustness gap apply broadly to certification tasks in which the
operational data model may deviate from the idealized one. This
perspective is particularly relevant for realistic quantum
communication and measurement settings with adaptive classical
control, efficiency constraints, and environmental drift.

\section{Related work}

Semi-device-independent (SDI) certification has been widely studied,
particularly in the context of dimension witnesses and prepare-and-measure
protocols.
Device-independent tests of classical versus quantum dimensionality
were introduced in~\cite{Gallego2010DIM}, and the SDI framework was
further developed for cryptographic tasks, including prepare-and-measure
quantum key distribution, in~\cite{Pawlowski2011SDI}.
These works establish classical performance bounds under specified
operational assumptions and connect random access codes to security
guarantees.

Finite-size effects and martingale-based concentration techniques
have been developed extensively in both SDI and fully device-independent
scenarios, providing statistically sound certification tools
under well-defined models.
Such analyses focus on deriving confidence intervals and asymptotic
security guarantees once the operational assumptions are fixed.

Detection inefficiencies and postselection effects have also been
investigated.
In device-independent Bell tests, detection loopholes and fair-sampling
assumptions are well understood~\cite{pearle1970hidden,larsson1998bell,brunner2014bell}.

Beyond the standard detection loophole, Bell tests have also provided
clear precedents for the role of postselection, coincidence windows,
and temporal dependence in certification claims.
The coincidence-time loophole shows that setting-dependent pairing or
selection of events can produce apparent Bell violations unless the
inequality is reformulated consistently with the actual event-selection
rule~\cite{LarssonGill2004,Larsson2014Loopholes,Gebhart2022}.
Likewise, the memory loophole and its statistical treatment clarified
that history-dependent local models do not invalidate Bell tests per se,
but they do require finite-sample analyses that remain valid beyond the
IID setting~\cite{BarrettMemory2002,ElkoussWehner2016,Zhang2013}.
More recently, the effect of postselection on device-independent claims
under fair-sampling assumptions has been analysed explicitly in the Bell
framework~\cite{Orsucci2020,Gebhart2023}.
In SDI settings, the impact of detection inefficiencies and selection
effects has been analysed explicitly, including detection-loophole
attacks on prepare-and-measure protocols~\cite{DallArno2012SDI}.
Related work has studied robustness of dimension witnesses to
experimental imperfections~\cite{DallArno2012DIM} and alternative
binary-outcome constructions~\cite{Czechlewski2018DIW}.

Recent work has also revisited detection efficiency thresholds in semi-device-independent and prepare-and-measure scenarios from a causal and operational perspective, including the role of nondetection modelling in dimension witnessing and related SDI tasks~\cite{Sarubi:2025ekp}. Such analyses study how imperfections modify the admissible correlation set or the witness threshold itself. Our focus is complementary: we do not analyse detector inefficiency as a physical loophole per se, but rather the consistency condition that the score, finite-sample bound, and classical benchmark must be defined under the same operational data model.

In most prior treatments, both in Bell and in SDI settings,
operational deviations are addressed either at the level of the
physical model---by deriving modified local or classical bounds under
altered assumptions---or at the level of statistical analysis---by
establishing finite-sample guarantees once the operational model is
fixed.

In particular, the possibility that misalignment between these
elements can generate asymptotically persistent false certification
has not been isolated as a structural modelling vulnerability.
The present work complements existing theory by introducing a
protocol-agnostic robustness diagnostic that integrates scoring,
concentration, and benchmark definition within a single operational
alignment framework. 

Unlike standard loophole analyses in Bell tests, or modified-threshold
analyses in SDI settings, which primarily ask how operational
imperfections alter the admissible correlation set or the relevant
witness bound, the present work isolates a modelling-layer effect
arising from the interaction between scoring, statistical
concentration, and benchmark definition.

\section{Certification under operational assumptions}
\label{sec:cert_model}

\subsection{Score and witness estimation}

Consider a certification task defined by a performance witness $S$ computed
from experimental data. We assume that in each round $t=1,\dots,N$, an
experiment produces a bounded outcome variable $X_t \in [0,1]$, where
$X_t=1$ corresponds to a successful event and $X_t=0$ to failure.
The empirical score is defined as
\begin{equation}
    \hat S = \frac{1}{N_{\mathrm{eval}}}
    \sum_{t \in \mathcal{E}} X_t ,
\end{equation}
where $\mathcal{E}$ denotes the set of evaluation rounds and
$N_{\mathrm{eval}} = |\mathcal{E}|$.

In idealized analyses, $\mathcal{E}$ typically coincides with all
experimental rounds. However, in practical settings $\mathcal{E}$ may
depend on selection rules (e.g.\ discarding certain rounds). The
definition of $\mathcal{E}$ is therefore part of the certification
protocol and must be specified explicitly.

\subsection{Finite-sample confidence bounds}

Certification statements are made at finite data and must therefore
incorporate statistical uncertainty. We define a lower confidence bound
$S_{\mathrm{low}}$ satisfying
\begin{equation}
    \Pr\!\left( S_{\mathrm{true}} \ge S_{\mathrm{low}} \right)
    \ge 1-\alpha ,
\end{equation}
for a chosen confidence level $1-\alpha$.

In the bounded-increment setting $X_t \in [0,1]$, martingale
concentration inequalities provide finite-sample guarantees without
requiring independence assumptions. A representative example is the
Azuma--Hoeffding bound,
\begin{equation}
    S_{\mathrm{low}}
    = \hat S -
    \sqrt{\frac{\ln(1/\alpha)}{2 N_{\mathrm{eval}}}} ,
\end{equation}
which is valid under bounded conditional expectations.
The use of martingale-safe bounds is particularly appropriate when
temporal correlations or adaptive strategies may be present.

Other martingale inequalities, such as variance-adaptive Bernstein or Freedman-type bounds, may provide tighter finite-sample constants when conditional variance information is available; the alignment principle developed here is agnostic to the specific concentration tool, provided the bound remains valid under the operational filtration.

\subsection{Classical reference values}

Let $S_C$ denote the classical ceiling under the ideal operational
model assumed by the protocol. In practice, deviations in the
data-generating process may alter the set of admissible classical
strategies and hence modify the appropriate reference value.
We denote by $S_{C,\mathrm{eff}}$ the \emph{effective classical ceiling}
corresponding to the actual operational model.

Importantly, $S_{C,\mathrm{eff}}$ must be derived under the same
assumptions that govern the observed data (e.g.\ biased inputs,
memory effects, or selection rules). Comparing $\hat S$ to an
idealized $S_C$ when the operational model differs can lead to
misleading conclusions.

\subsection{Acceptance criterion}

A certification protocol accepts a supra-classical claim only if
\begin{equation}
    S_{\mathrm{low}} > S_{C,\mathrm{ref}},
\end{equation}
where $S_{C,\mathrm{ref}}$ is the classical benchmark used for
comparison. Robust certification requires that
$S_{C,\mathrm{ref}} = S_{C,\mathrm{eff}}$, i.e.\ that the benchmark
matches the operational data model.

The robustness gap introduced in Sec.~\ref{sec:intro},
\begin{equation}
    \Delta_{\mathrm{rob}}
    = S_{\mathrm{low}} - S_{C,\mathrm{eff}},
\end{equation}
provides a quantitative measure of this condition.

\begin{proposition}[Effective classical ceiling as a finite optimisation]
\label{prop:effective_ceiling}
Consider a finite prepare-and-measure task with preparation input
$a\in\mathcal A$, measurement setting $y\in\mathcal Y$, output
$b\in\mathcal B$, and linear score
\begin{equation}
S=\sum_{a,y,b} c_{ayb}\, p(b|a,y).
\label{eq:linear_score_prop}
\end{equation}
Assume that the preparation device is restricted to transmit a
classical message $m\in\mathcal M$ with $|\mathcal M|\le d$, and that
the operational input law $\pi(a,y)$ is fixed. Then the set of
classically achievable behaviours $p(b|a,y)$ is a convex polytope whose
extreme points are the deterministic encoder--decoder strategies
\[
f:\mathcal A\to\mathcal M,
\qquad
g:\mathcal M\times\mathcal Y\to\mathcal B.
\]
Consequently, the effective classical ceiling is
\begin{equation}
S_{C,\mathrm{eff}}
=
\max_{\lambda\in\Lambda_{\mathrm{det}}} S_\lambda,
\label{eq:SCeff_vertex}
\end{equation}
where $\Lambda_{\mathrm{det}}$ denotes the finite set of deterministic
encoder--decoder pairs and $S_\lambda$ is the expected score of the
strategy $\lambda$ under the operational law $\pi(a,y)$.
Equivalently, $S_{C,\mathrm{eff}}$ can be obtained as a linear program
over convex weights on deterministic strategies.
\end{proposition}

\begin{proof}
For fixed finite sets $\mathcal A,\mathcal Y,\mathcal B,\mathcal M$,
there are only finitely many deterministic encoders
$f:\mathcal A\to\mathcal M$ and deterministic decoders
$g:\mathcal M\times\mathcal Y\to\mathcal B$.
Each pair $\lambda=(f,g)$ defines a deterministic conditional behaviour
\[
p_\lambda(b|a,y)=\delta_{b,g(f(a),y)}.
\]
Allowing shared randomness means that any classical strategy is a convex
mixture of such deterministic behaviours:
\begin{equation}
p(b|a,y)=\sum_{\lambda\in\Lambda_{\mathrm{det}}}
q_\lambda\, p_\lambda(b|a,y),
\qquad
q_\lambda\ge 0,
\qquad
\sum_{\lambda} q_\lambda=1.
\label{eq:convex_mixture_det}
\end{equation}
Hence the classical behaviour set is the convex hull of finitely many
points, and therefore a convex polytope.

Since the score in Eq.~\eqref{eq:linear_score_prop} is linear in
$p(b|a,y)$, it is also linear in the mixture weights $q_\lambda$:
\begin{equation}
S=\sum_{\lambda\in\Lambda_{\mathrm{det}}} q_\lambda S_\lambda,
\qquad
S_\lambda :=
\sum_{a,y,b} c_{ayb}\,\pi(a,y)\, p_\lambda(b|a,y).
\end{equation}
A linear functional over a convex polytope attains its maximum at an
extreme point. Therefore the maximum classical score is attained by at
least one deterministic strategy, yielding
Eq.~\eqref{eq:SCeff_vertex}. The equivalent linear-program formulation
follows directly from Eq.~\eqref{eq:convex_mixture_det}.
\end{proof}

This result shows that the operational classical benchmark is not an
ad hoc correction, but the solution of a well-defined optimisation
problem determined by the same assumptions that govern the data.

\paragraph{Example: biased $2\to 1$ RAC.}
In the $2\to 1$ random access code with $y\in\{0,1\}$ drawn from a biased distribution $\Pr(y=0)=q$, the classical message alphabet has size $|\mathcal M|=2$. Enumerating deterministic encoders and decoders shows that the optimal strategy guesses perfectly the more probable query and randomly the less probable one, yielding
\begin{equation}
S_{C,\mathrm{eff}}(q)=\frac{1+\max(q,1-q)}{2}
=\frac{3}{4}+\frac{|2q-1|}{4}.
\end{equation}
Writing $\varepsilon=2q-1$, this becomes
\begin{equation}
S_{C,\mathrm{eff}}(\varepsilon)=\frac{3}{4}+\frac{|\varepsilon|}{2},
\end{equation}
in agreement with the analytic ceiling used in Sec.~6.4.

\section{Operational model deviations}
\label{sec:deviations}

The certification framework described in Sec.~\ref{sec:cert_model}
assumes an operational model under which the classical reference
$S_C$ is defined. In realistic implementations, deviations from
this model may arise. These deviations do not necessarily imply
a breakdown of the underlying physical mechanism; rather, they
modify the effective classical performance that must be used
as reference.

We frame the certification problem in the general prepare-and-measure (PAM) semi-device-independent (SDI) setting. In each trial, a referee samples a preparation input $a\in\mathcal A$ and a measurement setting $y\in\mathcal Y$, and sends $a$ to the preparation device and $y$ to the measurement device. The preparation device emits a physical system or classical message $m$, constrained by a bounded resource (for instance, a classical alphabet $|\mathcal M|\le d$ or a quantum system of Hilbert-space dimension at most $d$). The measurement device produces an output $b\in\mathcal B$. The experiment is summarized by conditional probabilities $p(b|a,y)$ and by a linear task score
\begin{equation}
S=\sum_{a,y,b} c_{ayb}\, p(b|a,y).
\end{equation}
Realistic implementations deviate from the ideal PAM model in several ways. In this work we focus on three operational deviations: biased setting generation, temporal correlations or adaptivity due to device memory, and selection or postselection rules that keep or discard trials. Our central requirement is that the score definition, the finite-sample statistical guarantee applied to that score, and the classical benchmark used for comparison must all be derived under the same operational data-generating model. Figure~\ref{fig:pam_dag} summarizes the PAM structure together with the operational deviations considered in this work.

\begin{figure}[ht!]
\centering
\begin{tikzpicture}[
    >=Latex,
    thick,
    every node/.style={font=\normalsize},
    obs/.style={circle, draw, minimum size=12mm, fill=green!65, inner sep=0pt},
    op/.style={circle, draw, minimum size=12mm, fill=violet!70, inner sep=0pt},
    hid/.style={rounded rectangle, draw, minimum width=11mm, minimum height=10mm, fill=yellow!80!orange},
    dsh/.style={->, dashed},
    sol/.style={->}
]

\node[obs] (x) at (0,1.8) {$a$};
\node[op]  (st) at (3.2,1.8) {$S_t$};
\node[obs] (y) at (6.4,1.8) {$y$};
\node[obs] (r) at (8.5,1.8) {$R$};

\node[op]  (m) at (1.1,0) {$m$};
\node[op]  (b) at (5.5,0) {$b$};

\node[hid] (lam) at (3.3,-1.7) {$\lambda$};
\node[op]  (kt) at (6.7,-1.8) {$K_t$};

\draw[sol] (x) -- (m);
\draw[sol] (m) -- (b);
\draw[sol] (y) -- (b);
\draw[sol] (lam) -- (m);
\draw[sol] (lam) -- (b);

\draw[dsh] (st) -- (m);
\draw[dsh] (st) -- (b);
\draw[dsh] (r) -- (y);
\draw[dsh] (kt) -- (b);
\draw[dsh] (kt) -- (m);

\end{tikzpicture}
\caption{\textbf{Prepare-and-measure causal diagram with operational deviations.}
Green nodes denote observed inputs, violet nodes operational or internal variables, and the yellow node a latent shared variable. Solid arrows indicate the ideal PAM structure, while dashed arrows represent operational deviations such as memory/adaptivity, imperfect setting generation, and selection/postselection.}
\label{fig:pam_dag}
\end{figure}
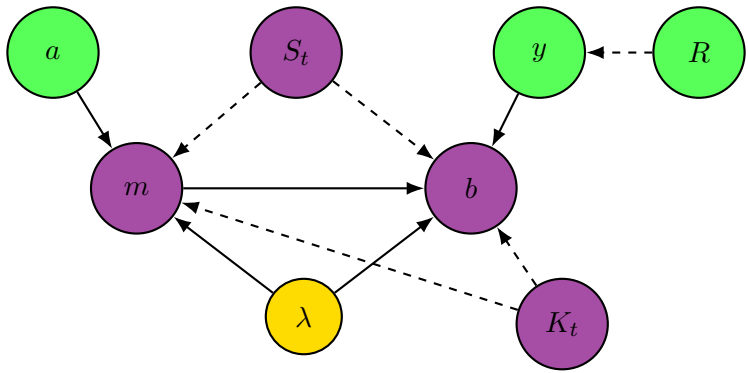
\subsection{Input bias}

Many certification protocols assume uniformly distributed input
settings. In practice, input distributions may be biased,
either unintentionally (imperfect random number generation)
or as a consequence of system drift.

Let $\varepsilon$ parameterize a deviation from uniform inputs.
Such bias can enlarge the set of effective classical strategies,
since a classical device may optimize performance toward the
most frequently queried settings. The appropriate classical
reference value must therefore be replaced by a bias-dependent
ceiling $S_{C,\mathrm{eff}}(\varepsilon)$.

Failure to adjust the classical benchmark accordingly can lead
to apparent supra-classical scores that merely reflect the
mismatch between assumed and actual input statistics.

\subsection{Temporal correlations and memory}

Ideal certification analyses often assume independence across
rounds. However, temporal correlations may arise due to memory
effects in the device or correlations in the input generation
process.

When correlations are present, adaptive classical strategies can
exploit information from previous rounds to improve performance.
In such cases, the effective classical ceiling must be computed
under the correlated operational model. Statistical guarantees
must also remain valid under non-IID data, motivating the use
of martingale-based concentration bounds as described in
Sec.~\ref{sec:cert_model}.

\subsection{Selection effects and postselection}

A particularly delicate deviation arises when the set of
evaluation rounds $\mathcal{E}$ is determined by a selection
rule, such as discarding certain rounds (e.g.\ detection failures
or data-cleaning steps).

If the empirical score is computed conditionally on the kept
rounds,
\begin{equation}
    \hat S_{\mathrm{cond}}
    = \frac{1}{|\mathcal{E}|}
      \sum_{t \in \mathcal{E}} X_t,
\end{equation}
then the resulting statistic does not represent the average
success probability over all experimental attempts.
In particular, selective discarding of unfavorable outcomes can
artificially inflate $\hat S_{\mathrm{cond}}$.

This phenomenon is structurally analogous to detection-efficiency
loopholes in other certification contexts: conditioning changes
the meaning of the witness.

By contrast, an unconditional score that normalizes over all
test rounds preserves the operational interpretation of the
success probability and prevents artificial inflation through
postselection.

\subsection{Adaptive classical strategies}

Finally, classical devices may adapt their behavior based on
observed data, especially when input distributions or correlations
are structured. Learning-based strategies provide a systematic
way to explore such adaptive behaviors.

Importantly, the presence of adaptive classical strategies does
not by itself imply a failure of certification. Rather, it
highlights the need to define $S_{C,\mathrm{eff}}$ under the
same operational assumptions that govern the data-generating
process. When scoring rules, statistical bounds, and classical
references are properly aligned, adaptive strategies can at most
reconstruct the effective classical ceiling implied by the
operational model.

\section{Robust certification via alignment}
\label{sec:robust}

The deviations discussed in Sec.~\ref{sec:deviations} reveal a common
structural issue: certification claims depend not only on the observed
data, but on the consistency between scoring rules, statistical bounds,
and classical reference values. We formalize this requirement as an
\emph{alignment principle}.

\subsection{Alignment principle}

\begin{definition}[Operational alignment]
A certification protocol is operationally aligned if:
\begin{enumerate}
    \item The score $\hat S$ estimates the performance quantity
    corresponding to the full set of operational trials.
    \item The confidence bound $S_{\mathrm{low}}$ remains valid under
    the actual data-generating process (including correlations or
    adaptive behavior).
    \item The classical reference value $S_{C,\mathrm{ref}}$ coincides
    with the effective classical ceiling $S_{C,\mathrm{eff}}$ implied
    by the same operational model.
\end{enumerate}
\end{definition}

Violation of any of these conditions can lead to artificial
supra-classical claims.
Equivalently, an operational model induces a triplet
(score definition, statistical bound, classical reference set).
Operational alignment requires that all three components
be derived under the same assumptions on the data-generating process.

\subsection{Unconditional scoring}

A key structural requirement concerns the definition of the score.
Let $X_t \in [0,1]$ denote the  unconditional
success indicator per attempted test round $t$. In the presence
of postselection, this corresponds to replacing $X_t$ by
$\chi_t(H_t) X_t$, so that discarded rounds contribute zero.

We define the unconditional score
\begin{equation}
    \hat S
    = \frac{1}{N_{\mathrm{test}}}
      \sum_{t \in \mathcal{T}} X_t,
\end{equation}
where $\mathcal{T}$ is the set of all designated test rounds and
$N_{\mathrm{test}} = |\mathcal{T}|$.

Under this definition, rounds that are discarded by an internal
selection rule contribute effectively as failures if no valid
success outcome is recorded. This preserves the operational meaning
of $\hat S$ as the success probability per attempted test round.

By contrast, conditional scoring over a subset of ``kept'' rounds
modifies the denominator and changes the interpretation of the
witness, opening the possibility of score inflation through
selection effects.

Note that in the unconditional setting considered throughout the remainder
of the paper, we identify $T$ with the set of all designated test
rounds and write $N_{\mathrm{test}} = N$ for simplicity.

\subsection{Martingale-safe bounds under deviations}

When temporal correlations or adaptive strategies are present,
independence assumptions may fail. However, as long as the increments
$X_t$ remain bounded in $[0,1]$, martingale concentration inequalities
provide valid finite-sample guarantees.

Using such bounds ensures that $S_{\mathrm{low}}$ remains valid
even under memory effects or adaptive classical strategies.
This removes the need to assume IID behavior in order to control
statistical fluctuations.

\subsection{Robustness gap}

We now formalize the central diagnostic quantity.

\begin{definition}[Robustness gap]
Given an operational model,
the robustness gap is defined as
\begin{equation}
    \Delta_{\mathrm{rob}}
    = S_{\mathrm{low}} - S_{C,\mathrm{eff}} .
\end{equation}
\end{definition}

A certification claim is supported at confidence level
$1-\alpha$ only if $\Delta_{\mathrm{rob}} > 0$.
Conversely, $\Delta_{\mathrm{rob}} \le 0$ indicates that the data
are compatible with classical performance under the operational
model.

\subsection{False certification under misalignment}

Artificial certification can arise in two structurally distinct ways:

\begin{enumerate}
    \item \emph{Statistical fluctuation}: $\hat S$ exceeds
    $S_{C,\mathrm{eff}}$ due to finite-sample noise.
    This is controlled by the confidence bound $S_{\mathrm{low}}$.
    \item \emph{Model misalignment}: $\hat S$ is compared to an
    incorrect reference value $S_C \neq S_{C,\mathrm{eff}}$,
    or is computed using a score that does not represent the
    operational success probability.
\end{enumerate}

The first mechanism is intrinsic to finite data and is suppressed
exponentially with increasing $N_{\mathrm{test}}$.
The second is structural and can persist even asymptotically unless
the scoring rule and classical benchmark are properly aligned. The first mechanism is suppressed asymptotically as
$N_{\mathrm{test}} \to \infty$,
whereas the second can persist unless the protocol is
operationally aligned.

The robustness gap cleanly separates these mechanisms by combining
finite-sample control with an operationally consistent reference.

\subsection{Adversarial classical strategies}
\label{sec:adversary}

We now formalize the notion of an adversarial classical strategy
within the operational framework.

Consider a sequence of rounds $t=1,\dots,N$.
In each round, the device receives inputs $I_t$ (e.g.\ preparation
settings or queries) and produces an output $O_t$.
A classical strategy is specified by a conditional probability rule
\begin{equation}
    P(O_t \mid I_t, H_{t-1}),
\end{equation}
where $H_{t-1}$ denotes the full history of prior inputs and outputs
up to round $t-1$.

\paragraph{Memoryless strategies.}
If $P(O_t \mid I_t, H_{t-1}) = P(O_t \mid I_t)$,
the device is memoryless and rounds are independent.

\paragraph{Correlated strategies.}
If the distribution depends on $H_{t-1}$ only through a finite
internal state $S_{t-1}$ evolving as
\begin{equation}
    S_t = f(S_{t-1}, I_t, O_t),
\end{equation}
the device implements a finite-memory strategy.

\paragraph{Fully adaptive strategies.}
More generally, the device may update its internal parameters
based on the entire observed history, potentially implementing
a learning rule. Such strategies remain classical as long as
the outputs are generated by classical stochastic processes.

\medskip

An \emph{adversarial classical strategy} is any strategy within
this class that is optimized to maximize the expected certification
score under the given operational model.

Importantly, this definition is algorithm-independent.
Reinforcement learning or other adaptive algorithms provide
constructive implementations of such strategies, but do not
expand the set of achievable classical correlations.

\paragraph{Implications for certification.}
Under model deviations such as biased inputs or temporal
correlations, adversarial classical strategies may achieve
performance exceeding the ideal classical ceiling $S_C$.
However, their performance remains bounded by the effective
ceiling $S_{C,\mathrm{eff}}$ derived under the same operational
assumptions.

Consequently, adaptive learning does not constitute a new
physical loophole; rather, it serves as a systematic method
to explore the classical strategy space permitted by the
operational model.

 \subsection{Adversarial postselection}
\label{sec:postselection_formal}

In addition to choosing outputs, an operational implementation may
determine whether a given round is included in the evaluation set.
We model this through a selection function
\begin{equation}
    \chi_t(H_t) \in \{0,1\},
\end{equation}
where $H_t$ denotes the full history up to and including round $t$,
and $\chi_t=1$ indicates that round $t$ is kept for evaluation.

The effective evaluation set is therefore
\begin{equation}
    \mathcal{E}
    = \{ t \in \mathcal{T} : \chi_t(H_t) = 1 \},
\end{equation}
where $\mathcal{T}$ denotes the designated test rounds.

\paragraph{Memoryless selection.}
If $\chi_t$ depends only on local variables (e.g.\ detection events),
the selection process is memoryless.

\paragraph{Adaptive selection.}
More generally, $\chi_t$ may depend arbitrarily on the history $H_t$,
including previous successes and failures. In this case the selection
rule itself may be viewed as an adaptive classical strategy.

\medskip

An \emph{adversarial postselection rule} is any history-dependent
function $\chi_t(H_t)$ designed to maximize the reported score
under a given scoring rule.

\paragraph{Impact on scoring.}
If the empirical score is defined conditionally as
\begin{equation}
    \hat S_{\mathrm{cond}}
    = \frac{\sum_{t \in \mathcal{E}} X_t}
           {|\mathcal{E}|},
\end{equation}
then adversarial postselection can increase $\hat S_{\mathrm{cond}}$
by preferentially discarding unfavorable outcomes.
In the extreme case, selective discarding of all failures yields
$\hat S_{\mathrm{cond}} = 1$ independently of the underlying
data-generating process.

By contrast, the unconditional score
\begin{equation}
    \hat S
    = \frac{1}{N_{\mathrm{test}}}
      \sum_{t \in \mathcal{T}}
      \chi_t(H_t) X_t
\end{equation}
retains the interpretation of success probability per attempted
test round, since discarded rounds effectively contribute as
failures. Under this definition, adversarial postselection cannot
inflate the score beyond what is achievable under the operational
model.

\paragraph{Robust certification with postselection.}
When unconditional scoring is combined with martingale-safe
confidence bounds, the lower bound $S_{\mathrm{low}}$ remains
valid even if both the output strategy and the selection rule
are adaptive and history-dependent.

\section{Case study: the $2\!\to\!1$ random access code}
\label{sec:rac}

We now illustrate the framework in the minimal
$2\!\to\!1$ random access code (RAC), which serves as a controlled
testbed for SDI certification.

\subsection{Semi-device-independent certification}

Semi-device-independent (SDI) certification refers to protocols in
which performance guarantees are derived under partial structural
assumptions about the devices, without requiring full characterization
of their internal workings.

In contrast to fully device-independent (DI) scenarios, which rely on
Bell nonlocality and assume only no-signaling constraints, SDI
approaches typically impose a bounded resource assumption. A common
example is a dimension bound: the communicated system between
preparation and measurement devices is restricted to a fixed
dimension (e.g.\ a single classical or quantum bit).

Under such constraints, one can derive classical performance bounds
for specific tasks. Exceeding these bounds certifies the use of
non-classical resources consistent with the assumed dimension limit.

The $2\!\to\!1$ random access code (RAC) provides a minimal example.
If only one classical bit is communicated, the maximal classical
success probability under uniform inputs is $S_C=\tfrac34$.
Exceeding this value certifies the use of non-classical (quantum)
resources within the assumed communication dimension.

The SDI character of the protocol arises from the fact that no
assumptions are made about the internal implementation of the
devices beyond the communication constraint. Certification therefore
depends critically on the validity of the operational assumptions
under which the classical benchmark is derived.

\subsection{Setup}

In each round:

\begin{itemize}
    \item A preparation device receives two uniformly random bits
    $a_0,a_1 \in \{0,1\}$.
    \item A measurement device receives a query bit $y \in \{0,1\}$.
    \item The preparation device sends a single classical bit $m$.
    \item The measurement device outputs $b \in \{0,1\}$.
\end{itemize}

The success event is
\begin{equation}
    X_t = \mathbf{1}\!\left[ b_t = a_{y_t} \right].
\end{equation}

The empirical score is $\hat S = \frac{1}{N} \sum_t X_t$ under
unconditional scoring.

\subsection{Ideal classical ceiling}

Under uniformly random $y$, any classical strategy can send
only one of the two bits. Without loss of generality,
consider a deterministic strategy that always sends $a_0$.
Then
\begin{equation}
    \Pr(b=a_y)
    =
    \Pr(y=0) \cdot 1
    +
    \Pr(y=1) \cdot \tfrac12.
\end{equation}

For uniform inputs $\Pr(y=0)=\Pr(y=1)=\tfrac12$, yielding
\begin{equation}
    S_C = \tfrac34.
\end{equation}

No classical strategy exceeds this value.

\subsection{Biased inputs and effective ceiling}

Suppose now that the query distribution is biased:
\begin{equation}
    \Pr(y=0) = \tfrac12 + \varepsilon.
\end{equation}

The optimal classical strategy is to send the bit corresponding
to the more frequently queried input.
Assuming $\varepsilon \ge 0$, the optimal choice is to send $a_0$,
leading to
\begin{align}
    S_{C,\mathrm{eff}}(\varepsilon)
    &=
    \left(\tfrac12 + \varepsilon\right)\cdot 1
    +
    \left(\tfrac12 - \varepsilon\right)\cdot \tfrac12
    \\
    &=
    \tfrac34 + \tfrac{\varepsilon}{2}.
\end{align}

For $\varepsilon<0$, the symmetric expression yields
\begin{equation}
    S_{C,\mathrm{eff}} = \tfrac34 + \tfrac{|\varepsilon|}{2}.
\end{equation}

Thus input bias enlarges the classical ceiling linearly in
$|\varepsilon|$.

\subsubsection{Bias uncertainty and robust classical ceilings}
\label{subsec:bias_uncertainty}

In practice, the setting bias may not be known exactly due to finite-sample fluctuations, calibration uncertainty, or slow drift. To obtain certification statements that remain valid under such uncertainty, it is natural to assume only that the bias parameter satisfies
\begin{equation}
|\varepsilon|\le \varepsilon_{\max},
\end{equation}
and to define the worst-case robust classical ceiling
\begin{equation}
S_{C,\mathrm{rob}}(\varepsilon_{\max})
=
\sup_{|\varepsilon|\le \varepsilon_{\max}} S_{C,\mathrm{eff}}(\varepsilon).
\label{eq:SCrob}
\end{equation}
For the $2\to 1$ RAC, using
\begin{equation}
S_{C,\mathrm{eff}}(\varepsilon)=\frac34+\frac{|\varepsilon|}{2},
\end{equation}
we obtain
\begin{equation}
S_{C,\mathrm{rob}}(\varepsilon_{\max})
=
\frac34+\frac{\varepsilon_{\max}}{2}.
\label{eq:SCrob_RAC}
\end{equation}
A robust supra-classical claim is therefore accepted only if
\begin{equation}
S_{\mathrm{low}} > S_{C,\mathrm{rob}}(\varepsilon_{\max}),
\end{equation}
or equivalently if the minimax robustness gap
\begin{equation}
\Delta_{\mathrm{rob}}^{\mathrm{minimax}}
=
S_{\mathrm{low}} - S_{C,\mathrm{rob}}(\varepsilon_{\max})
\end{equation}
is strictly positive.

\begin{proposition}[Worst-case robust ceiling under bounded bias uncertainty]
\label{prop:robust_ceiling_bias}
Suppose the operational bias parameter satisfies
\begin{equation}
|\varepsilon|\le \varepsilon_{\max},
\end{equation}
and let $S_{C,\mathrm{eff}}(\varepsilon)$ denote the effective classical
ceiling for fixed bias $\varepsilon$. Then the appropriate worst-case
classical benchmark is
\begin{equation}
S_{C,\mathrm{rob}}(\varepsilon_{\max})
=
\sup_{|\varepsilon|\le \varepsilon_{\max}}
S_{C,\mathrm{eff}}(\varepsilon).
\label{eq:robust_ceiling_general}
\end{equation}
Any certification claim based only on the bound
$|\varepsilon|\le\varepsilon_{\max}$ is valid at confidence level
$1-\alpha$ only if
\begin{equation}
S_{\mathrm{low}} > S_{C,\mathrm{rob}}(\varepsilon_{\max}).
\label{eq:robust_acceptance_general}
\end{equation}
For the biased $2\to1$ RAC, where
\begin{equation}
S_{C,\mathrm{eff}}(\varepsilon)=\frac34+\frac{|\varepsilon|}{2},
\end{equation}
one obtains
\begin{equation}
S_{C,\mathrm{rob}}(\varepsilon_{\max})
=
\frac34+\frac{\varepsilon_{\max}}{2}.
\label{eq:robust_ceiling_rac_prop}
\end{equation}
\end{proposition}

\begin{proof}
If only the interval $|\varepsilon|\le\varepsilon_{\max}$ is known,
then a certification statement must hold uniformly over all operational
models compatible with that information. Therefore the relevant
classical benchmark is the supremum of the effective ceiling over the
allowed bias interval, yielding Eq.~\eqref{eq:robust_ceiling_general}.
A claim of supra-classical performance is valid only if the lower
confidence bound exceeds this worst-case classical benchmark, which
gives Eq.~\eqref{eq:robust_acceptance_general}.

For the $2\to1$ RAC, $S_{C,\mathrm{eff}}(\varepsilon)$ is monotone in
$|\varepsilon|$, so the supremum over $|\varepsilon|\le\varepsilon_{\max}$
is attained at $|\varepsilon|=\varepsilon_{\max}$, yielding
Eq.~\eqref{eq:robust_ceiling_rac_prop}.
\end{proof}
\paragraph{Estimating $\varepsilon_{\max}$ from data.}
Let $q=\Pr(y=0)$ and let $\hat q=n_0/N$ be the observed frequency of the setting $y=0$ over $N$ trials. A $(1-\beta)$ confidence interval for $q$ can be obtained, for example, from Hoeffding's inequality:
\begin{equation}
\Pr\!\left(|\hat q-q|\le \delta\right)\ge 1-\beta,
\qquad
\delta=\sqrt{\frac{\ln(2/\beta)}{2N}}.
\end{equation}
Since $\varepsilon=2q-1$, this implies
\begin{equation}
|\varepsilon-\hat\varepsilon|\le 2\delta,
\qquad
\hat\varepsilon=2\hat q-1.
\end{equation}
A conservative operational choice is therefore
\begin{equation}
\varepsilon_{\max}=|\hat\varepsilon|+2\delta,
\end{equation}
which yields a fully data-driven robust benchmark through Eq.~\eqref{eq:SCrob_RAC}.

\subsection{Conditional scoring and postselection inflation}

Consider now a selection rule $\chi_t(H_t)$ that discards
a fraction of rounds. If the score is computed conditionally as
\begin{equation}
    \hat S_{\mathrm{cond}}
    =
    \frac{\sum_{t\in\mathcal{E}} X_t}
         {|\mathcal{E}|},
\end{equation}
an adversarial selection rule may discard primarily failure
events.

In the extreme case where all failures are discarded,
\begin{equation}
    \hat S_{\mathrm{cond}} = 1,
\end{equation}
independently of the underlying strategy.

Comparing $\hat S_{\mathrm{cond}}$ to $S_C=\tfrac34$
would then yield an artificial violation.

This illustrates that conditional scoring changes the
operational meaning of the witness and invalidates
the comparison with the ideal classical ceiling.

\subsection{Unconditional scoring restores robustness}

Under unconditional scoring,
\begin{equation}
    \hat S
    =
    \frac{1}{N}
    \sum_{t=1}^N \chi_t(H_t) X_t,
\end{equation}
discarded rounds effectively contribute as failures.

In this case, selective discarding cannot increase the
expected score beyond the effective classical ceiling
$S_{C,\mathrm{eff}}$ implied by the operational model.

\subsection{Learning-based adversarial stress tests}

To probe whether adaptive classical strategies can exploit
operational deviations in nontrivial ways, we implement
learning-based classical agents using standard reinforcement
learning (RL) algorithms \cite{sutton2018reinforcement}.
Their role is purely diagnostic: they provide an automated search
over the classical strategy space under biased, correlated,
or drifting input distributions. They are not proposed as a new
certification method, but as stress tests for the alignment
framework.

We consider simple bandit-based learners and windowed
state-dependent variants that update action-value estimates
from observed rewards.
These algorithms adapt the preparation strategy in response
to empirical input statistics, thereby tracking changes in
the operational model.
Importantly, the receiver applies the fixed classical decoding
rule $b=m$, so that all adaptive behaviour is confined to
the preparation device.

The learned strategies improve empirical performance relative
to static classical policies when input statistics drift or
are biased.
However, this improvement reflects adaptation to the effective
classical optimisation problem defined by the operational model.
When compared against the corresponding effective classical
ceiling $S_{C,\mathrm{eff}}$, the robustness gap $\Delta_{\mathrm{rob}}$ remains
non-positive for all classical agents.

Thus, learning-based adversaries do not enlarge the classical
correlation set.
They serve as automated search procedures that reconstruct
the operationally implied classical bound.
False certification arises from benchmarking under an idealised
reference that does not reflect the operational model,
rather than from adaptive optimisation itself.

\subsection{Numerical results}

We illustrate the preceding analysis with simulations of the
$2\!\to\!1$ RAC under representative operational deviations.
The numerical results are designed to highlight the central logic of the
paper: apparent supra-classical behaviour can arise from benchmark
misalignment, whereas aligned scoring and benchmarking restore a correct
classical interpretation.

\begin{figure}[ht!]
  \centering
  \includegraphics[width=0.86\linewidth]{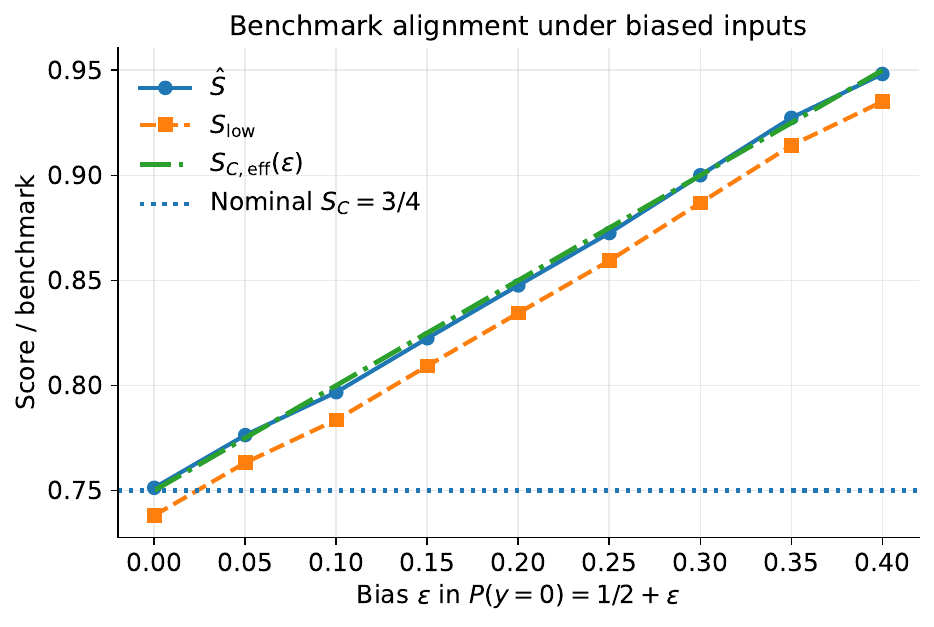}
  \caption{
  Benchmark alignment under biased inputs.
  The empirical score $\hat S$ and its lower confidence bound
  $S_{\mathrm{low}}$ are shown as functions of the input bias
  $\varepsilon$, together with the nominal classical benchmark
  $S_C=3/4$ and the bias-corrected effective ceiling
  $S_{C,\mathrm{eff}}(\varepsilon)=3/4+|\varepsilon|/2$.
  When compared against the nominal benchmark, the data would appear
  to violate the classical limit. Once the operationally correct
  benchmark is used, the same behaviour is seen to be entirely
  classical.
  }
  \label{fig:alignment_bias}
\end{figure}

\paragraph{Input bias and benchmark alignment.}
Figure~\ref{fig:alignment_bias} provides the clearest numerical
illustration of the alignment principle.
As the input bias increases, the observed score $\hat S$ rises together
with the martingale-safe lower bound $S_{\mathrm{low}}$.
If one compares these data to the nominal SDI threshold
$S_C=3/4$, the behaviour would be misread as increasingly
supra-classical.
However, the correct operational comparison is with the effective
classical ceiling
\[
S_{C,\mathrm{eff}}(\varepsilon)=\frac34+\frac{|\varepsilon|}{2},
\]
which tracks the attainable classical performance under the same
biased input law.
The figure therefore makes explicit that false certification can arise
purely from benchmark misalignment, even when the observed data are
fully classical.

\begin{figure}[ht!]
  \centering
  \includegraphics[width=0.86\linewidth]{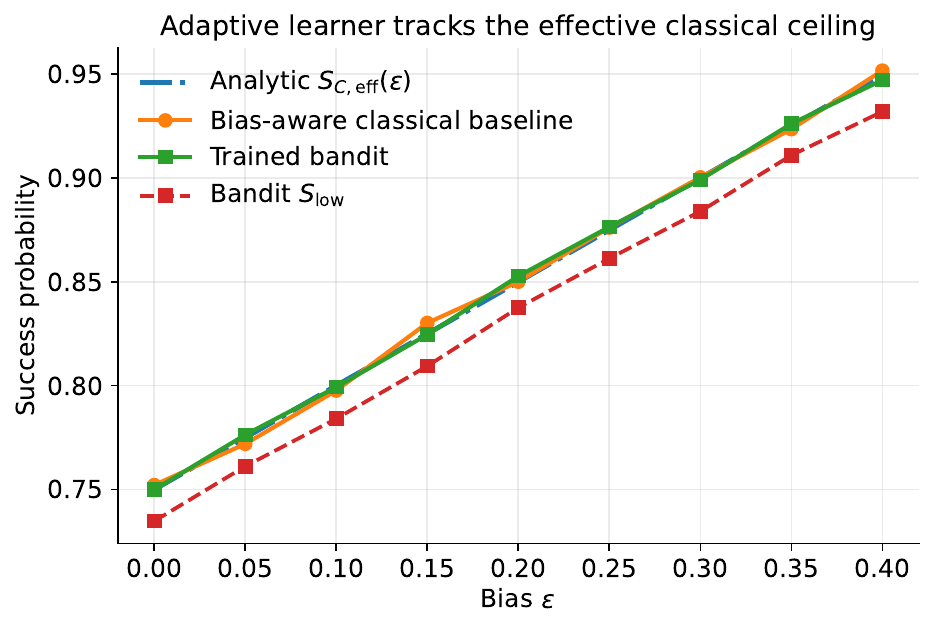}
  \caption{
  Adaptive classical learning recovers the effective classical ceiling.
The analytically optimal benchmark $S_{C,\mathrm{eff}}(\varepsilon)$ is
compared with a bias-aware classical baseline, a trained bandit
strategy, and the lower confidence bound for the learned policy.
The learned agent tracks the same operationally correct classical
benchmark and does not exceed it.
  }
  \label{fig:learner_ceiling}
\end{figure}

\paragraph{Adaptive strategies recover the effective classical ceiling.}
Figure~\ref{fig:learner_ceiling} clarifies the role of learning-based
classical adversaries in the present framework.
The trained bandit does not reveal a new loophole or enlarge the
classical correlation set; rather, it learns to approximate the same
bias-dependent optimum already captured analytically by
$S_{C,\mathrm{eff}}(\varepsilon)$.
Its performance closely follows both the analytic ceiling and the
bias-aware optimal classical baseline, while its lower confidence bound
remains below the ceiling.
This supports the interpretation of adaptive learning as a diagnostic
search procedure over the allowed classical strategy space, not as a
distinct source of supra-classical behaviour.

\begin{figure}[t]
  \centering
  \includegraphics[width=0.86\linewidth]{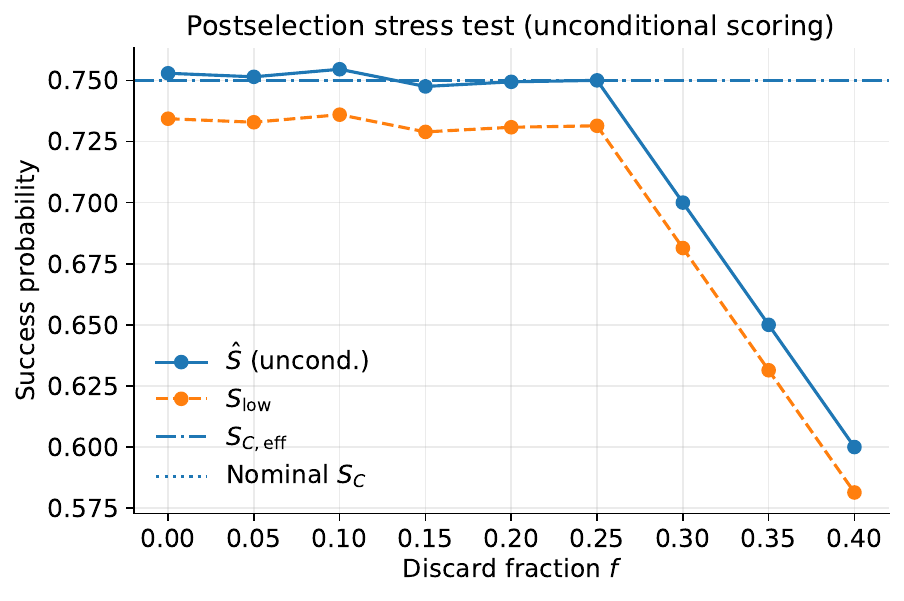}
  \caption{
  Postselection stress test.
  Conditional scoring leads to inflated apparent success probabilities
  and spurious certification, whereas unconditional scoring restores
  compatibility with the effective classical ceiling.
  }
  \label{fig:postselection}
\end{figure}

\paragraph{Postselection stress test.}
Figure~\ref{fig:postselection} isolates the second main mechanism of
false certification considered in this work.
When the score is computed conditionally on kept rounds only,
adversarial postselection can preferentially discard failures and
artificially inflate the reported performance.
This produces apparent violations of the nominal classical threshold
that do not correspond to any genuine non-classical resource.
By contrast, unconditional scoring assigns discarded rounds zero
contribution and preserves the operational interpretation of the score
as success probability per attempted test round.
Under this aligned definition, compatibility with the effective
classical ceiling is restored.

\paragraph{Robustness-gap interpretation.}
Across all these numerical tests, the relevant diagnostic is not the
empirical score alone but the robustness gap
\[
\Delta_{\mathrm{rob}} = S_{\mathrm{low}} - S_{C,\mathrm{eff}}.
\]
Under aligned certification, this quantity remains non-positive for all
classical strategies, including adaptive and history-dependent ones.
Apparent violations arise only when the benchmark or scoring rule is
defined inconsistently with the operational data model.\paragraph{Adaptive strategies and nonstationarity.}
In the unified operational scenario,
including mild nonstationarity of the input statistics,
adaptive classical strategies (e.g.\ bandit-based learners \cite{auer2002finite})
recover at most the effective classical ceiling.
As shown in the right plot in Fig.~\ref{fig:memory_misalignment},
the robustness gap remains negative for all classical policies,
indicating that adaptive learning explores the classical strategy
space without producing supra-classical behaviour.

\subsection{Memory amplifies benchmark misalignment}

In contrast to the aligned setting of Fig.~\ref{fig:learner_ceiling}, we now consider a deliberately misaligned evaluation rule in which conditional scoring and a nominal benchmark are used. In this regime, adaptive memory does not enlarge the classical set, but it can amplify the rate of false certification by exploiting operational drift more efficiently than static strategies.

To isolate this effect,
we consider a moderate operational deviation regime in which
static classical strategies do not systematically trigger
false certification when compared against the nominal bound
$S_C = 0.75$, but adaptive strategies can exploit the same
deviations more effectively.
\begin{figure}[ht!]
  \centering
  \includegraphics[width=0.95\linewidth]{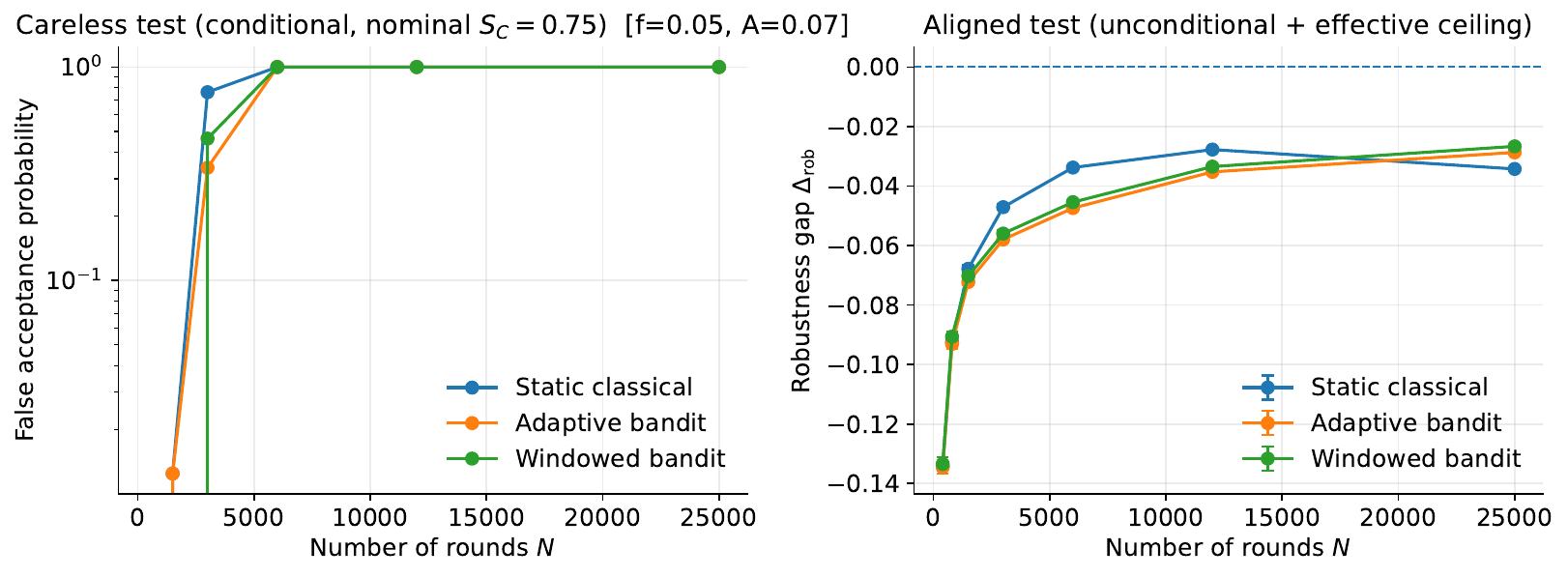}
  \caption{
  Memory amplifies benchmark misalignment.
  Left: false-acceptance probability under careless evaluation
  (conditional scoring and nominal classical bound).
  Adaptive strategies exhibit substantially higher false-positive
  rates than static policies.
  Right: robustness gap under aligned certification
  (unconditional scoring and effective classical ceiling).
  All classical strategies satisfy $\Delta_{\mathrm{rob}}\le 0$.
  }
  \label{fig:memory_misalignment}
\end{figure}

Figure~\ref{fig:memory_misalignment} shows the false-acceptance
probability under a careless evaluation rule
(conditional scoring and nominal benchmark) as a function of
the number of rounds $N$.
In this regime, the static classical strategy yields only
occasional false positives, whereas adaptive bandit-based
strategies produce substantially higher false-acceptance rates.
The improvement reflects the ability of memoryful policies
to track temporal drift and bias in the input distribution,
thereby approaching the operationally implied classical ceiling
more efficiently.

The right panel of Fig.~\ref{fig:memory_misalignment}
shows the corresponding robustness gap under aligned
certification.
When the effective classical ceiling $S_{C,\mathrm{eff}}$
is computed consistently with the operational model and
unconditional scoring is enforced, all classical strategies—
including adaptive ones—satisfy
$\Delta_{\mathrm{rob}} \le 0$ and converge to zero from below
as $N$ increases. The parameter regime is chosen such that the operational
deviation is moderate: static strategies do not deterministically
violate the nominal benchmark, while adaptive strategies
track drift sufficiently well to do so with high probability.

This comparison highlights the structural point:
adaptive classical memory does not enlarge the admissible
correlation set, but it amplifies the consequences of
benchmark misalignment.
False certification arises from inconsistent modelling,
not from adaptive optimisation itself.

\section{Discussion and outlook}
\label{sec:discussion}
The central result of this work is that false certification can be
structural rather than statistical. Operational deviations such as
bias, drift, memory, and postselection do not by themselves invalidate
SDI protocols; rather, they modify the effective classical
optimisation problem against which performance must be evaluated.
When score, concentration bound, and classical benchmark are defined
consistently with the operational data-generating process, classical
strategies --- including adaptive and memoryful ones --- remain below
the appropriate ceiling. By contrast, inconsistent benchmarking can
generate persistent false certification, even in the large-sample
limit.

The $2\to1$ RAC was chosen as a transparent testbed, not as the scope
of the claim. What matters is the alignment principle itself and the
associated robustness gap. We now discuss the broader implications of
that principle, its relevance beyond the specific example considered
here, and directions for further development.

\subsection{Protocol-agnostic lessons}

The central message of this work is independent of the specific
RAC example. Certification relies not only on statistical
concentration, but on the alignment between the operational model
and the classical benchmark against which performance is compared.
When input statistics, efficiency constraints, or memory effects
modify the effective classical optimisation problem, the nominal
bound $S_C$ no longer provides a valid reference.

The alignment principle introduced here formalises this requirement.
By defining an effective classical ceiling $S_{C,\mathrm{eff}}$
consistent with the actual operational assumptions, and by
quantifying deviations through the robustness gap
\[
\Delta_{\mathrm{rob}} = S_{\mathrm{low}} - S_{C,\mathrm{eff}},
\]
one obtains a diagnostic that separates statistical fluctuations
from structural modelling errors. 

The construction is
protocol-agnostic: it applies to any witness-based certification
scenario in which the relevant classical reference can be computed
under the same assumptions as the operational data model. In finite
prepare-and-measure scenarios, Proposition~\ref{prop:effective_ceiling}
makes this statement explicit by showing that the operational classical
benchmark is the solution of a well-defined optimisation over the convex
hull of deterministic classical strategies.
\subsection{Relevance for quantum technologies and cyber-relevant protocols}

In practical quantum-communication and cryptographic deployments,
certification is embedded within classical control and
post-processing layers. Quantum key distribution (QKD), randomness
expansion, and semi-device-independent (SDI) protocols rely on
classical sampling, basis-choice selection, parameter estimation,
error correction, and privacy amplification. These steps help define
the operational model against which security or quantumness claims
are evaluated.

Real-world implementations introduce additional structure:
biased sampling, efficiency filtering, detector dead times,
time-dependent drift, and adaptive control logic. In networked or
software-mediated settings, classical preprocessing may also include
anomaly detection, data curation, or real-time adjustment of protocol
parameters in response to estimated channel conditions. Such
mechanisms alter the distribution of inputs and outputs and therefore
modify the relevant classical benchmark.

If the benchmark is not updated to reflect the actual operational
constraints, classical effects compatible with the deployment model
may be incorrectly interpreted as supra-classical behaviour.
Importantly, this is not a cryptographic vulnerability of the
underlying protocol, but a modelling vulnerability in the
certification layer. Defining an effective classical ceiling
consistent with the true operational assumptions restores alignment
and closes this gap.

\subsection{Asymptotic misalignment}

Statistical uncertainty vanishes in the large-sample limit,
but structural misalignment does not.
If the classical benchmark $S_C$ does not match the operational
model (e.g.\ biased or correlated inputs), then even as
$N \to \infty$ the empirical score $\hat S$ converges to an
incorrect reference, potentially leading to persistent false
certification.

By contrast, when the effective classical ceiling
$S_{C,\mathrm{eff}}$ is defined consistently with the actual
input and device model, the robustness gap
\[
\Delta_{\mathrm{rob}} = S_{\mathrm{low}} - S_{C,\mathrm{eff}}
\]
converges to zero from below for all classical strategies.
Finite-size fluctuations disappear asymptotically, while
misalignment errors remain unless explicitly corrected.
This distinction cleanly separates statistical effects from
structural modelling assumptions in certification.

\subsection{Limitations and next steps}

This work provides a framework and a minimal testbed rather than a comprehensive survey of certification protocols. We do not derive new classical bounds or security proofs; rather, we provide a consistency layer that can be applied to existing certification frameworks.

The $2\to1$ RAC was chosen for transparency and analytic tractability, but the alignment principle is not tied to this specific witness. Future work should extend the analysis to other semi-device-independent and fully device-independent scenarios, including entropic witnesses, higher-dimensional dimension tests, and Bell-type inequalities. An important practical direction is to combine operationally aligned concentration bounds with algorithmic computation of effective classical ceilings over the relevant classical strategy polytope, and with robust minimax benchmarks when operational parameters such as setting bias are only known within confidence intervals. It will also be important to study richer device models incorporating memory, detector inefficiencies, and experimentally measured drift, as well as to apply the robustness-gap diagnostic to real datasets. More broadly, systematic integration of certification logic with adaptive classical control and data-curation pipelines remains an open and practically relevant direction for emerging quantum technologies.

\section*{Acknowledgements}
We would like to thank Manuel Gessner for interesting conversations on the use of RL in quantum device certifications.
The work
of V.S. is supported by the Spanish grants PID2023-148162NB-C21 and CEX2023-001292-S. 
The work of A.S. is supported by the National Natural Science Foundation of China (Grant No.~W2531008) and the Peacock Plan.


\appendix

\section{Concentration bounds used in this work}
\label{app:bounds}

We briefly state the concentration inequality used to define the
lower confidence bound $S_{\mathrm{low}}$.

Let $\{X_t\}_{t=1}^{N}$ be binary random variables
$X_t \in \{0,1\}$ representing the success indicator in each
test round.
Define the empirical mean
\[
\hat S = \frac{1}{N} \sum_{t=1}^{N} X_t .
\]

We assume that the increments are bounded,
\[
0 \le X_t \le 1,
\]
and that the process satisfies the martingale difference condition
with respect to the natural filtration
$\mathcal{F}_t$ generated by past outcomes,
\[
\mathbb{E}[X_t \mid \mathcal{F}_{t-1}] = \mu_t,
\]
where $\mu_t$ is the (possibly history-dependent)
conditional expectation under the classical model. Since $0 \le X_t \le 1$, the martingale differences
$X_t - \mathbb{E}[X_t \mid \mathcal{F}_{t-1}]$
are uniformly bounded, ensuring validity of the bound.

Under these assumptions, Azuma--Hoeffding's inequality yields
\[
\Pr\!\left( \hat S - \mathbb{E}[\hat S] \ge \epsilon \right)
\le \exp\!\left(-2 N \epsilon^2 \right).
\]

For a chosen significance level $\alpha$, we therefore define
the lower confidence bound
\[
S_{\mathrm{low}}
=
\hat S
-
\sqrt{\frac{\ln(1/\alpha)}{2N}}.
\]

All statistical statements in the main text are based on this
finite-sample bound.

The Azuma--Hoeffding bound is used here as a minimal martingale-safe choice requiring only bounded increments. In settings where one can also control or estimate the predictable quadratic variation, tighter finite-sample results may be obtained from Freedman- or Bernstein-type inequalities. Such refinements do not modify the alignment logic developed in the main text; they only sharpen the statistical component of the certification triplet.
\section{Details of the simulated operational deviations}
\label{app:sim_details}

This appendix specifies the operational deviations and adaptive
strategies used in the numerical simulations.

\subsection*{Input models}

\paragraph{Biased IID inputs.}
In the stationary bias model, the query bit $y_t \in \{0,1\}$
is sampled independently with
\[
\Pr(y_t=0) = \tfrac12 + \varepsilon,
\qquad
\Pr(y_t=1) = \tfrac12 - \varepsilon,
\]
with $|\varepsilon| \le 1/2$.

\paragraph{Biased Markov inputs.}
In the correlated-input model, $\{y_t\}$ forms a
first-order stationary Markov chain with transition probabilities
\[
\Pr(y_t = y_{t-1}) = p_{\mathrm{stay}},
\qquad
\Pr(y_t \neq y_{t-1}) = 1 - p_{\mathrm{stay}}.
\]
The stationary distribution is chosen to match a target bias
$\varepsilon$, so that
$\Pr(y_t=0)=\tfrac12+\varepsilon$.

\paragraph{Drift model.}
To model mild nonstationarity, the bias is allowed to vary slowly
in time, e.g.
\[
\varepsilon_t
=
\varepsilon_0
+
A \sin\!\left( \frac{2\pi t}{T} \right),
\]
or via a bounded random walk.
The query $y_t$ is then sampled with
$\Pr(y_t=0)=\tfrac12+\varepsilon_t$ independently
conditioned on $\varepsilon_t$.

\subsection*{Postselection model}

Postselection is implemented via a selection function
$\chi_t \in \{0,1\}$.
Given a target discard fraction $f$, a fraction $fN$
of rounds is removed.
In the adversarial model used in the stress tests,
discarded rounds are chosen preferentially among failures
whenever possible.
Conditional scoring computes the success rate over kept rounds only,
while unconditional scoring assigns zero contribution to discarded
rounds, preserving bounded increments.

\subsection*{Adaptive classical strategies}

\paragraph{Simple bandit.}
We consider the standard $2\!\to\!1$ RAC scenario, where the
preparation device receives input bits $(a_0,a_1)$ and sends a
single classical message $m \in \{0,1\}$ to the measurement
device, which receives a query $y \in \{0,1\}$ and outputs
a bit $b \in \{0,1\}$.

The preparation device selects between two classical encoding
actions:
\begin{itemize}
\item Action $0$: send $a_0$,
\item Action $1$: send $a_1$.
\end{itemize}
A two-armed bandit maintains action-value estimates $Q(a)$.
At round $t$, the action is chosen via an $\varepsilon$-greedy rule:
with probability $\varepsilon_{\mathrm{greedy}}$
a random action is selected, otherwise
$a = \arg\max_{a'} Q(a')$.
After observing reward $r_t \in \{0,1\}$,
the update rule is
\[
Q(a) \leftarrow Q(a) + \eta \bigl(r_t - Q(a)\bigr),
\]
with learning rate $\eta$.

\paragraph{Windowed bandit (regime detection).}
To model adaptive tracking under drift,
a state variable $s_t$ is defined from a sliding window of
the previous $W$ queries,
\[
s_t = 
\begin{cases}
0, & \hat{\varepsilon}_t \ge 0,\\
1, & \hat{\varepsilon}_t < 0,
\end{cases}
\]
where $\hat{\varepsilon}_t$ is the empirical bias estimate
over the window.
Separate action values $Q(s,a)$ are maintained and updated
as above.

Unless otherwise stated, simulations use
learning rate $\eta \approx 0.05$--$0.1$,
exploration probability
$\varepsilon_{\mathrm{greedy}} \approx 0.05$,
and window size $W$ between $100$ and $400$.
Results are robust to moderate variations of these parameters.

In all simulations, the measurement device performs the trivial
classical decoding $b = m$, so that the adaptive behaviour is
entirely confined to the preparation (encoding) strategy.
\bibliographystyle{unsrt}
\bibliography{references}

\end{document}